\documentclass{LMCS}

\def\dOi{11(1:12)2014}
\lmcsheading%
{\dOi}
{1--16}
{}
{}
{Jan.~10, 2014}
{Mar.~25, 2014}
{}

\ACMCCS{[\textbf{Theory of computation}]: Logic - Proof theory;
  Computational complexity and cryptography - Problems, reductions and
  completeness}

\usepackage{amsmath}
\usepackage{amsfonts}
\usepackage{amssymb}
\usepackage{cite}
\usepackage{hyperref}

\newcommand{\gneg}{\neg} 
\newcommand{\ada}{\mbox{\Large $\sqcap$}} 
\newcommand{\ade}{\mbox{\Large $\sqcup$}} 

\rm

\begin{document}

\title{The Computational Complexity of Propositional Cirquent Calculus}

\author[M.~S.~Bauer]{Matthew S. Bauer}	
\address{University of Illinois at Urbana-Champaign, U.S.A.}	
\email{msbauer2@illinois.edu} 


\begin{abstract}
Introduced in 2006 by Japaridze, cirquent calculus is a refinement of sequent calculus. The advent of cirquent calculus arose from the need for a deductive system with a more explicit ability to reason about resources. Unlike the more traditional proof-theoretic approaches that manipulate tree-like objects (formulas, sequents, etc.), cirquent calculus is based on circuit-style structures called {\em cirquents}, in which different “peer” (sibling, cousin, etc.) substructures may share components. It is this resource sharing mechanism to which cirquent calculus owes its novelty (and its virtues). From its inception, cirquent calculus has been paired with an abstract resource semantics. This semantics allows for reasoning about the interaction between a resource provider and a resource user, where resources are understood in the their most general and intuitive sense. Interpreting resources in a more restricted computational sense has made cirquent calculus instrumental in axiomatizing various fundamental fragments of \emph{Computability Logic}, a formal theory of (interactive) computability. The so-called ``classical" rules of cirquent calculus, in the absence of the particularly troublesome \emph{contraction} rule, produce a sound and complete system CL5 for Computability Logic. In this paper, we investigate the computational complexity of CL5, showing it is $\Sigma_2^p$-complete. We also show that CL5 without the \emph{duplication} rule has polynomial size proofs and is NP-complete. 
\end{abstract}

\maketitle

\section{Introduction}

Introduced in 2006 by Japaridze \cite{ICC}, cirquent calculus is a refinement of classical sequent calculus. The advent of cirquent calculus arose from the need for a deductive system with a more explicit ability to reason about resources. Unlike the more traditional proof-theoretic approaches that manipulate tree-like objects (formulas, sequents, etc.), cirquent calculus is based on circuit-style structures called {\em cirquents}, in which different “peer” (sibling, cousin, etc.) substructures may share components. It is this resource sharing mechanism to which cirquent calculus owes its novelty (and its virtues). Cirquents come in a variety of forms. One way to characterize the sort of cirquents studied in this paper in familiar terms is to say that a cirquent is a multiset of sequents (called “groups”) where each formula --- more precisely, each occurrence of a formula --- may simultaneously belong to more than one sequent. This explains the origin of the word “cirquent”, which is a hybrid of “circuit” and “sequent”. From its inception, cirquent calculus has been paired with an abstract resource semantics. This semantics allows for reasoning about the interaction between a resource provider and a resource user, where resources are understood in the their most general and intuitive sense. 

Interpreting resources in a more restricted computational sense has made cirquent calculus instrumental in axiomatizing various fundamental fragments of \emph{Computability Logic}, a formal theory of (interactive) computability. The so-called ``classical" rules of cirquent calculus, in the absence of the particularly troublesome \emph{contraction} rule, produce a sound and complete system CL5 for Computability Logic. Born in \cite{ICLI}, Computability Logic (CoL) is an ambitious research program aimed at developing a formal theory of interactive computability. To this end, formulas of CoL represent computational problems modeled at games. The notion of ``truth" for such formulas becomes synonymous with the notion of ``computability" for the computational problems they represent.      

While CL5 has shown to be sound and complete with respect Japridze's abstract resource semantics, it has also shown to validate a strictly larger class of formulas than affine logic (sequent calculus without the contraction rule), the latter being well studied and sound as a logic of resources. However, due to no shortcomings in effort, a complete resource aware semantics has not been found for affine logic. This had lead Japridze to conclude that ``CL5 rather than affine logic adequately materializes the resource philosophy traditionally associated with the latter". Indeed, the ``semantics before syntax" philosophy upon which CL5 was conceived, together with its completeness result provide compelling evidence in this direction. 

In parallel with the development of cirquent calculus, recent work \cite{PCLI, PCLII, FTCI, FTCII, TRI} has introduced sound and complete axiomizations for several fragments of CoL, the most fundamental of these being CL4, which contains the propositional connectives $\neg$ (negation), $\vee$ (parallel disjunction), $\wedge$ (parallel conjunction), $\sqcup$ (choice disjunction) and $\sqcap$ (choice conjunction) as well as the ``choice" quantifiers $\ade$ and $\ada$ and the ``blind" quantifiers $\forall$ and $\exists$. For a full discussion of these operators and the semantics of CoL, see \cite{GS}. While the semantics of CoL differs drastically from that of classical logic, the set of valid formulas of CL4 is identical to that of classical logic when restricted to its operators and the so-called ``elementary" sort of atoms. Thus, the presence of the quantifiers $\forall$ and $\exists$ within CL4's language guarantees its undecidability. In \cite{FTCII}, however, the $\forall$,$\exists$-free fragment of CL4 was shown to be decidable in polynomial space. A PSPACE-completeness proof was later given in \cite{MSB}.

The deductive apparatus of CL4 is far removed from traditional Gentzen or Hilbert style axiomizations. Further, the system does not extend very naturally to other fragments of Computability Logic. A rectification to this shortcoming came with the advent of cirquent calculus. As mentioned previously, the fragment of cirquent calculus not containing the  \emph{contraction} rule, produces a sound and complete system for the fragment of Computability Logic known as CL5. This system contains the operators $\vee$, $\wedge$ and $\neg$ and in this paper we show it to be $\Sigma_2^p$-complete.  

We further investigate the complexity of the logical system described by the rules of CL5, but with the absence of the \emph{duplication} rule. We call this logic CL5$^-$ and show that every theorem in CL5$^-$ has a polynomial size proof. This effectively places the problem of deciding CL5$^-$ provability in $NP$. A reduction from the $k$ vertex cover problem to CL5$^-$ provability is then given,  solidifying the logic as NP-complete.  In some sense, CL5 is the CoL ``counterpart" to multiplicative linear logic, the former being a proper extension of the latter. Here, CL5$^{-}$ lies strictly between multiplicative affine logic (linear logic with weakening) and CL5. Among the virtues of multiplicative affine logic is that, unlike the coNP-complete classical logic, it is in NP. CL5$^-$ thus presents a nice and natural extension of multiplicative affine logic that gets ``closer" to classical logic while still remaining in NP. 

With the tight relationship between cirquent calculus and Computability Logic, one should not overlook the potential impact of cirquent calculus on other areas of logic and proof theory. In \cite{IFL}, Xu constructs a cirquent calculus based system for the propositional fragment of independence friendly (IF) logic, allowing one to account for independence from propositional connectives in the same vein that traditional IF logic accounts for independence from quantifiers. In \cite{CCD} a cirquent calculus based proof system was developed which, among other benifits, yields polynomial size proofs for all instances of the pigeonhole principle. To date, this is the first proof system to achieve such a result without leveraging cut, extension or substitution.  

This paper is organized as follows. In section \ref{2}, we provide an introduction to the rules of cirquent calculus and describe the systems CL5 and CL5$^-$. Section \ref{3} studies the system CL5$^-$ and proves its NP-completeness. In the final section, we show the full CL5 to be $\Sigma_2^p$-complete by a reduction from the TQBF-$\Sigma_2$ problem.

\section{Core cirquent calculus}\label{2}

The present section provides the technical definitions and information necessary for understanding the results of this paper. For a full discussion of the concepts of cirquent calculus, see \cite{ICC, GS}. As noted previously, a \textbf{formula} of the language of cirquent calculus is syntactically identical to that of classical propositional logic, where formulas are built from the connectives $\neg$, $\wedge$ and $\vee$ as well as non-logical \textbf{atoms} (also known as propositional letters). The only difference between the two being that the language of cirquent calculus does not contain logical atoms such as $\top$ or $\perp$; the latter can simply be understood as abbreviations of $P \wedge \neg P$ and $P \vee \neg P$, respectively, for some (whatever) atom $P$. Our language further mandates that $\neg$ can only be applied to atoms. As always, a \textbf{literal} is an atom with or without the prefix $\neg$.   

A $k$-ary \textbf{pool}, for $k \in \mathbb{N}$, is a sequence Pl = $\langle F_1, F_2,...,F_k \rangle$ of formulas. The formulas in a pool are not required to be unique. We refer to a particular occurrence of a formula as an \textbf{oformula}. A $k$-ary \textbf{structure}, for $k \in \mathbb{N}$, is a finite sequence St = $\langle \Gamma_1, \Gamma_2,...,\Gamma_m \rangle$ for $m \geq 0$, where each $\Gamma_i$, called a \textbf{group} of St, is a subset of $\{1, 2,...,k\}$. It is permitted that $\Gamma_i = \Gamma_j$ for $i \neq j$ and hence we use the term \textbf{ogroup} to refer to a particular occurrence of a group. We are now ready to give our definition for a cirquent. 

\begin{defi} \cite{ICC}
A $k$-ary ($k \geq 0$) \textbf{cirquent} is a pair $C$ = (St$^C$, Pl$^C$), where St$^{C}$, called the structure of $C$, is a $k$-ary \textbf{structure}, and PL$^C$, called the \textbf{pool} of $C$, is a $k$-ary pool. 
\end{defi} 

For example, let Pl$^C$ =  $\langle A, B, A, D \rangle$ and St$^C$ = $\langle \{1, 2\}, \{2, 3\}, \{4\} \rangle$. Here, $C$ has 4 oformulas and 3 ogroups and is typically represented by the diagram below.

 
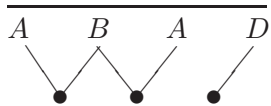
\begin{center} 
 \begin{picture}(98,54)
\put(0,44){\line(1,0){98}}
\put(0,32){$A$}
\put(30,32){$B$}
\put(60,32){$A$}
\put(90,32){$D$}

\put(20,10){\line(-2,3){13}}
\put(20,10){\line(4,5){15}}
\put(49,10){\line(-4,5){15}}
\put(49,10){\line(4,5){15}}
\put(78,10){\line(4,5){15}}
\put(20,10){\circle*{5}}
\put(49,10){\circle*{5}}
\put(78,10){\circle*{5}}
\end{picture}
\end{center}

Following \cite{ICC}, we will adopt the notion that a diagram simply ``is" (rather than just ``represents") a cirquent. When an ogroup $\Gamma$ is connected with an arc to an oformula $G$, we say that $\Gamma$ \textbf{contains} $G$. 

Before defining the rules for our inference system, let us first give some additional terminology necessary for the definition of our rules. For a given cirquent, two oformulas $F$ and $G$ are called  \textbf{adjacent} if $G$ is positioned immediately to the right of $F$.  In such a case, it is said that $F$ \textbf{immediately precedes} $G$ and $G$ \textbf{immediately follows} $F$.  \textbf{Merging two adjacent ogroups} $\Gamma$ and $\Delta$ in a given cirquent $C$ means replacing in $C$ the two ogroups $\Gamma$ and $\Delta$ by a single ogroup $\Gamma \cup \Delta$. The rightmost cirquent below represents the cirquent that results from merging the first and second ogroups in the leftmost cirquent. 

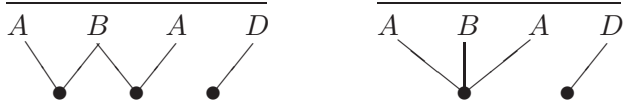
\begin{center} \begin{picture}(230,54)

\put(0,44){\line(1,0){98}}
\put(0,32){$A$}
\put(30,32){$B$}
\put(60,32){$A$}
\put(90,32){$D$}

\put(20,10){\line(-2,3){13}}
\put(20,10){\line(4,5){15}}
\put(49,10){\line(-4,5){15}}
\put(49,10){\line(4,5){15}}
\put(78,10){\line(4,5){15}}
\put(20,10){\circle*{5}}
\put(49,10){\circle*{5}}
\put(78,10){\circle*{5}}

\put(140,44){\line(1,0){92}}
\put(140,32){$A$}
\put(170,32){$B$}
\put(197,32){$A$}
\put(224,32){$D$}

\put(173,10){\line(-5,4){25}}
\put(173,10){\line(5,4){25}}
\put(173,10){\line(0,1){19}}
\put(212,10){\line(4,5){15}}
\put(173,10){\circle*{5}}
\put(212,10){\circle*{5}}
\end{picture}
\end{center}

\textbf{Merging two adjacent oformulas} $F$ and $G$ into a new oformula $H$ is the result of replacing $F$ an $G$ by $H$ and redirecting to it all of the arcs that were pointing to $F$ or $G$. The rightmost cirquent below represents the cirquent obtained from merging, in the leftmost cirquent, (the first) $A$ and $B$ into a single oformula $E$.  

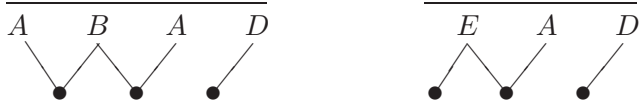
\begin{center} \begin{picture}(230,54)

\put(0,44){\line(1,0){98}}
\put(0,32){$A$}
\put(30,32){$B$}
\put(60,32){$A$}
\put(90,32){$D$}

\put(20,10){\line(-2,3){13}}
\put(20,10){\line(4,5){15}}
\put(49,10){\line(-4,5){15}}
\put(49,10){\line(4,5){15}}
\put(78,10){\line(4,5){15}}
\put(20,10){\circle*{5}}
\put(49,10){\circle*{5}}
\put(78,10){\circle*{5}}

\put(158,44){\line(1,0){80}}
\put(170,32){$E$}
\put(200,32){$A$}
\put(230,32){$D$}
\put(162,10){\line(2,3){12}}
\put(189,10){\line(-4,5){15}}
\put(189,10){\line(4,5){15}}
\put(218,10){\line(4,5){15}}
\put(162,10){\circle*{5}}
\put(189,10){\circle*{5}}
\put(218,10){\circle*{5}}
\end{picture}
\end{center}

With these definitions in mind, we are now equipped to present the so-called ``core" cirquent calculus rules. The rules we give here are only those relevant to our language and are by no means exhaustive. In its full generality, CoL encompasses numerous other logical operators all accompanied by a deep and meaningful semantics. As predicted earlier in \cite{ICC}, recent works have produced cirquent calculus axiomazations for logics containing some of the more powerful operators of Computability Logic. See, for example, \cite{TRI, TRII}.

\subsection{Axioms (A)}   
Axioms are rules with an empty set of premises. In cirquent calcus, they come in two forms: the \textbf{\emph{empty cirquent axiom}} and the \textbf{\emph{identity axiom}}. Both varieties of the axiom rule are illustrated below. It is important to note that the \emph{identity axiom} is actually a scheme of infinitely many axioms, as $F$ stands for an arbitrary formula.

\begin{center} \begin{picture}(209,76)
\put(0,63){\em empty cirquent axiom}
\put(28,46){\line(1,0){36}}
\put(69,44){\scriptsize A}

\put(150,63){\em identity axiom}
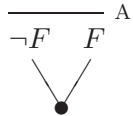
\put(160,46){\line(1,0){36}}
\put(200,44){\scriptsize A}
\put(160,33){$\gneg F$}
\put(188,33){$F$}
\put(180,10){\line(-3,5){11}}
\put(180,10){\line(3,5){11}}
\put(180,10){\circle*{5}}

\end{picture}
\end{center}

As will be our convention throughout, the letter placed to the right of the horizontal line represents the rule by which the conclusion was obtained.

\subsection{Mix (M)}

The \emph{mix} rule takes two premises. Its conclusion is obtained by placing each of the two premise cirquents side by side in a single cirquent, as illustrated below. 

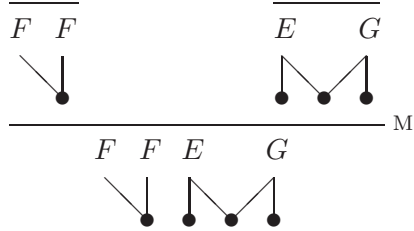
\begin{center} \begin{picture}(166,102)
\put(0,92){\line(1,0){26}}
\put(100,92){\line(1,0){40}}
\put(100,79){$E$}
\put(132,79){$G$}
\put(0,79){$F$}
\put(17,79){$F$}
\put(52,10){\circle*{5}}
\put(52,10){\line(0,1){16}}
\put(52,10){\line(-1,1){16}}
\put(68,10){\circle*{5}}
\put(84,10){\circle*{5}}
\put(100,10){\circle*{5}}
\put(68,10){\line(0,1){16}}
\put(84,10){\line(-1,1){16}}
\put(84,10){\line(1,1){16}}
\put(100,10){\line(0,1){16}}

\put(0,46){\line(1,0){142}}
\put(145,44){\scriptsize M}
\put(65,33){$E$}
\put(97,33){$G$}
\put(32,33){$F$}
\put(49,33){$F$}
\put(20,56){\circle*{5}}
\put(20,56){\line(0,1){16}}
\put(20,56){\line(-1,1){16}}
\put(103,56){\circle*{5}}
\put(119,56){\circle*{5}}
\put(135,56){\circle*{5}}
\put(103,56){\line(0,1){16}}
\put(119,56){\line(-1,1){16}}
\put(119,56){\line(1,1){16}}
\put(135,56){\line(0,1){16}}
\end{picture}
\end{center}

The remaining rules described in subsections \ref{exch} to \ref{cintro} all take a single premise cirquent. 

\subsection{Exchange (E)}\label{exch} 

The \emph{exchange} rule comes in two varieties, \textbf{\emph{oformula exchange}} and \textbf{\emph{ogroup exchange}}. The conclusion of the \emph{oformula exchange} rule is obtained by swapping the positions of two adjacent oformulas in the premise cirquent. The \emph{ogroup exchange} rule is similar in that it allows two adjacent ogroups in a premise cirquent to exchange positions in the conclusion. In both varieties of the rule, the arcs from each ogroup to its oformulas should be preserved. Below is an example of \emph{oformula} (resp. \emph{ogroup}) \emph{exchange} in which the oformulas F and G (resp. ogroups $\#$1 and $\#$2) are swapped.

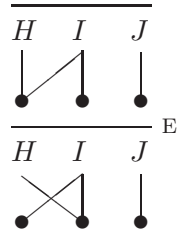
\begin{center} \begin{picture}(182,124)
\put(-4,109){\em oformula exchange}
\put(10,92){\line(1,0){53}}
\put(10,79){$E$}
\put(33,79){$F$}
\put(55,79){$G$}
\put(14,56){\line(0,1){19}}
\put(14,56){\line(5,4){23}}
\put(14,56){\circle*{5}}
\put(37,56){\circle*{5}}
\put(59,56){\circle*{5}}
\put(37,56){\line(0,1){18}}
\put(59,56){\line(0,1){18}}

\put(10,46){\line(1,0){53}}
\put(67,44){\scriptsize E}
\put(10,33){$E$}
\put(33,33){$G$}
\put(55,33){$F$}
\put(14,10){\line(0,1){18}}
\put(14,10){\line(5,2){45}}
\put(14,10){\circle*{5}}
\put(37,10){\circle*{5}}
\put(59,10){\circle*{5}}
\put(37,10){\line(0,1){18}}
\put(59,10){\line(0,1){18}}

\put(122,109){\em ogroup exchange}
\put(130,92){\line(1,0){53}}
\put(130,79){$H$}
\put(153,79){$I$}
\put(175,79){$J$}
\put(134,56){\line(0,1){19}}
\put(134,56){\line(5,4){23}}
\put(134,56){\circle*{5}}
\put(157,56){\circle*{5}}
\put(179,56){\circle*{5}}
\put(157,56){\line(0,1){19}}
\put(179,56){\line(0,1){18}}

\put(130,46){\line(1,0){53}}
\put(187,44){\scriptsize E}
\put(130,33){$H$}
\put(153,33){$I$}
\put(175,33){$J$}
\put(134,10){\line(5,4){23}}
\put(134,10){\circle*{5}}
\put(157,10){\circle*{5}}
\put(157,10){\line(-4,3){23}}
\put(157,10){\line(0,1){18}}
\put(179,10){\circle*{5}}
\put(179,10){\line(0,1){18}}
\end{picture}
\end{center}

\subsection{Weakening (W)} The \emph{weakening} rule has two forms; {\bf \emph{ogroup weakening}} and {\bf \emph{pool weakening}}. In the case of \emph{ogroup weakening}, the conclusion is obtained from the premise by adding a new arc between a pre-existing ogroup and oformula pair. In a application of \emph{pool weakening}, the conclusion is obtained by inserting a new oformula at any position in the pool of the premise. 

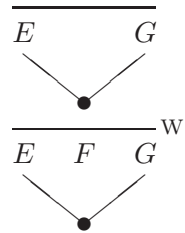
\begin{center} \begin{picture}(195,124)
\put(0,109){\em ogroup weakening}
\put(20,92){\line(1,0){31}}
\put(20,79){$E$}
\put(43,79){$F$}
\put(24,56){\line(0,1){18}}
\put(24,56){\circle*{5}}
\put(47,56){\circle*{5}}
\put(47,56){\line(0,1){18}}
\put(47,56){\line(-5,4){23}}

\put(20,46){\line(1,0){31}}
\put(55,44){\scriptsize W}
\put(20,33){$E$}
\put(43,33){$F$}
\put(24,10){\line(0,1){18}}
\put(24,10){\line(5,4){23}}
\put(24,10){\circle*{5}}
\put(47,10){\circle*{5}}
\put(47,10){\line(0,1){18}}
\put(47,10){\line(-5,4){23}}

\put(140,109){\em pool weakening}
\put(142,92){\line(1,0){54}}
\put(142,79){$E$}
\put(188,79){$G$}
\put(169,56){\circle*{5}}
\put(169,56){\line(-5,4){23}}
\put(169,56){\line(5,4){23}}

\put(142,46){\line(1,0){54}}
\put(198,44){\scriptsize W}
\put(142,33){$E$}
\put(165,33){$F$}
\put(188,33){$G$}
\put(169,10){\line(-5,4){23}}
\put(169,10){\line(5,4){23}}
\put(169,10){\circle*{5}}

\end{picture}
\end{center}

\subsection{Duplication (D)} 
The \emph{duplication} rule can be applied in one of two ways, called {\bf \emph{downward duplication}} and {\bf \emph{upward duplication}}. The conclusion of the \emph{downward duplication} rule is obtained from its premise by replacing an ogroup with two adjacent ogroups that each have arcs to exactly the same oformulas as the original ogroup. An application of \emph{upward duplication} works in the opposite direction in that the premise cirquent is obtained by replacing an ogroup in the conclusion by two adjacent ogroups that are both identical to the original ogroup.

\begin{center} \begin{picture}(238,122)
\put(0,107){\em downward duplication}
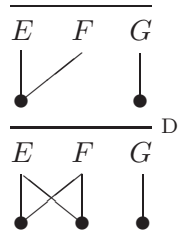
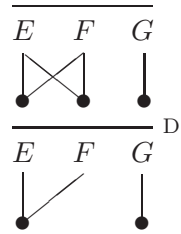
\put(19,92){\line(1,0){53}}
\put(19,79){$E$}
\put(42,79){$F$}
\put(64,79){$G$}
\put(23,56){\line(0,1){19}}
\put(23,56){\line(5,4){23}}
\put(23,56){\circle*{5}}
\put(68,56){\circle*{5}}
\put(68,56){\line(0,1){18}}

\put(19,46){\line(1,0){53}}
\put(76,44){\scriptsize D}
\put(19,33){$E$}
\put(42,33){$F$}
\put(64,33){$G$}
\put(23,10){\line(0,1){18}}
\put(23,10){\line(5,4){23}}
\put(23,10){\circle*{5}}
\put(46,10){\circle*{5}}
\put(69,10){\circle*{5}}
\put(46,10){\line(-5,4){22}}
\put(46,10){\line(0,1){18}}
\put(69,10){\line(0,1){18}}

\put(156,107){\em upward duplication}
\put(169,46){\line(1,0){53}}
\put(169,33){$E$}
\put(192,33){$F$}
\put(214,33){$G$}
\put(173,10){\line(0,1){19}}
\put(173,10){\line(5,4){23}}
\put(173,10){\circle*{5}}
\put(218,10){\circle*{5}}
\put(218,10){\line(0,1){18}}

\put(169,92){\line(1,0){53}}
\put(226,44){\scriptsize D}
\put(169,79){$E$}
\put(192,79){$F$}
\put(214,79){$G$}
\put(173,56){\line(0,1){18}}
\put(173,56){\line(5,4){23}}
\put(173,56){\circle*{5}}
\put(196,56){\circle*{5}}
\put(219,56){\circle*{5}}
\put(196,56){\line(-5,4){22}}
\put(196,56){\line(0,1){18}}
\put(219,56){\line(0,1){18}}

\end{picture}
\end{center}

\subsection{Contraction (C)} 
The \emph{contraction} rule takes a premise cirquent which contains two adjacent and identical oformulas. The conclusion of an application of this rule is obtained from the premise by merging two adjacent and identical oformulas $F$ and $F$ into a single oformula $F$.  

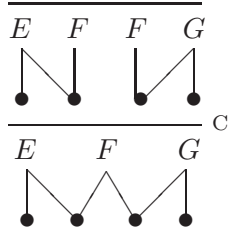
\begin{center} \begin{picture}(266,102)

\put(95,92){\line(1,0){73}}
\put(95,79){$E$}
\put(117,79){$F$}
\put(139,79){$F$}
\put(161,79){$G$}
\put(120,56){\line(0,1){19}}
\put(120,56){\circle*{5}}
\put(119,56){\line(-1,1){19}}
\put(146,56){\line(1,1){19}}
\put(145,56){\circle*{5}}
\put(165,56){\circle*{5}}
\put(165,56){\line(0,1){19}}
\put(100,56){\circle*{5}}
\put(100,56){\line(0,1){19}}
\put(143,56){\line(0,1){19}}

\put(95,46){\line(1,0){73}}
\put(172,44){\scriptsize C}
\put(97,33){$E$}
\put(128,33){$F$}
\put(159,33){$G$}
\put(121,10){\circle*{5}}
\put(143,10){\circle*{5}}
\put(121,10){\line(-1,1){19}}
\put(143,10){\line(1,1){19}}
\put(121,10){\line(3,5){11}}
\put(143,10){\line(-3,5){11}}
\put(102,10){\line(0,1){19}}
\put(102,10){\circle*{5}}
\put(162,10){\line(0,1){19}}
\put(162,10){\circle*{5}}

\end{picture}
\end{center}

\subsection{$\vee$-introduction ($\vee$)}  
The conclusion of this rule is obtained from the premise by merging two adjacent oformulas $F$ and $G$ into a single oformula $F \vee G$ such that all arcs pointing to either $F$ or $G$ now point to $F \vee G$. Two examples of an application of this rule are given below.

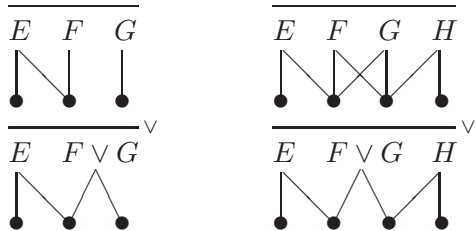
\begin{center} 
\begin{picture}(200,110)

\put(20,92){\line(1,0){49}}
\put(20,79){$E$}
\put(40,79){$F$}
\put(60,79){$G$}
\put(23,56){\circle*{5}}
\put(43,56){\circle*{5}}
\put(63,56){\circle*{5}}
\put(23,56){\line(0,1){19}}
\put(43,56){\line(0,1){19}}
\put(43,56){\line(-1,1){19}}
\put(63,56){\line(0,1){19}}

\put(20,46){\line(1,0){49}}
\put(71,44){\scriptsize $\vee$}

\put(20,33){$E$}
\put(40,33){$F \vee G$}
\put(23,10){\circle*{5}}
\put(43,10){\circle*{5}}
\put(63,10){\circle*{5}}
\put(23,10){\line(0,1){19}}
\put(43,10){\line(1,2){10}}
\put(43,10){\line(-1,1){19}}
\put(63,10){\line(-1,2){10}}

\put(120,92){\line(1,0){69}}
\put(120,79){$E$}
\put(140,79){$F$}
\put(160,79){$G$}
\put(180,79){$H$}
\put(123,56){\circle*{5}}
\put(143,56){\circle*{5}}
\put(163,56){\circle*{5}}
\put(123,56){\line(0,1){19}}
\put(143,56){\line(0,1){19}}
\put(143,56){\line(-1,1){19}}
\put(143,56){\line(1,1){19}}
\put(163,56){\line(0,1){19}}
\put(163,56){\line(-1,1){19}}
\put(163,56){\line(1,1){19}}
\put(184,56){\circle*{5}}
\put(183,56){\line(0,1){19}}

\put(120,46){\line(1,0){69}}
\put(191,44){\scriptsize $\vee$}
\put(120,33){$E$}
\put(140,33){$F\vee G$}
\put(180,33){$H$}
\put(123,10){\circle*{5}}
\put(143,10){\circle*{5}}
\put(164,10){\circle*{5}}
\put(123,10){\line(0,1){19}}
\put(143,10){\line(1,2){10}}
\put(143,10){\line(-1,1){19}}
\put(164,10){\line(-1,2){10}}
\put(164,10){\line(1,1){19}}
\put(184,10){\circle*{5}}
\put(183,10){\line(0,1){19}}

\end{picture}
\end{center}

\subsection{$\wedge$-introduction ($\wedge$)}\label{cintro} 

This rule takes a premise cirquent that contains two adjacent oformulas $F$ and $G$ such that no ogroup contains both $F$ and $G$ and every ogroup that contains $F$ (resp. $G$) is immediately followed (resp. preceded) by an ogroup containing $G$ (resp. $F$). The conclusion in an application of this rule is obtained from its premise by merging each ogroup that contains $F$ with the ogroup containing $G$ that immediately follows it. The oformulas $F$ and $G$ should then merge into $F \wedge G$. We again give two examples below. 

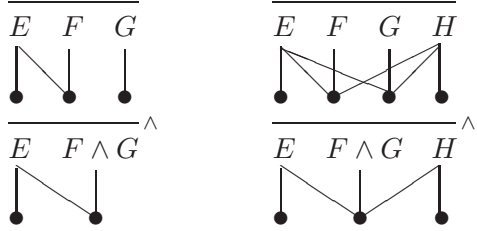
\begin{center} \begin{picture}(200,100)
\put(20,92){\line(1,0){49}}
\put(20,79){$E$}
\put(40,79){$F$}
\put(60,79){$G$}
\put(23,56){\circle*{5}}
\put(23,56){\line(0,1){19}}
\put(43,56){\circle*{5}}
\put(43,56){\line(0,1){18}}
\put(43,56){\line(-1,1){19}}
\put(64,56){\circle*{5}}
\put(64,56){\line(0,1){18}}

\put(20,46){\line(1,0){49}}
\put(71,44){\scriptsize $\wedge$}

\put(20,33){$E$}
\put(40,33){$F \wedge G$}
\put(23,10){\circle*{5}}
\put(23,10){\line(0,1){19}}
\put(53,10){\circle*{5}}
\put(53,10){\line(0,1){18}}
\put(53,10){\line(-3,2){30}}

\put(120,92){\line(1,0){69}}
\put(120,79){$E$}
\put(120,33){$E$}
\put(120,46){\line(1,0){69}}
\put(140,79){$F$}
\put(160,79){$G$}
\put(180,79){$H$}
\put(180,33){$H$}
\put(140,33){$F\wedge G$}
\put(123,56){\circle*{5}}

\put(123,56){\line(0,1){19}}
\put(143,56){\circle*{5}}
\put(143,56){\line(0,1){18}}
\put(142,56){\line(-1,1){19}}
\put(143,56){\line(2,1){40}}
\put(166,57){\line(-5,2){43}}
\put(164,56){\line(1,1){19}}
\put(164,56){\circle*{5}}
\put(164,56){\line(0,1){18}}
\put(183,56){\line(0,1){20}}
\put(184,56){\circle*{5}}

\put(191,44){\scriptsize $\wedge$}

\put(123,10){\circle*{5}}
\put(123,10){\line(0,1){19}}
\put(153,10){\circle*{5}}
\put(153,10){\line(0,1){18}}
\put(153,10){\line(-3,2){30}}
\put(153,10){\line(3,2){30}}
\put(184,10){\circle*{5}}
\put(183,10){\line(0,1){20}}

\end{picture}
\end{center}

\subsection{The systems CCC, CL5 and CL5$^-$}

The cirquent calculus system built from all eight of the above rules has been aptly named ``Classical Cirquent Calculus", or \textbf{CCC}. By removing the \emph{contraction} rule from CCC, we get the system \textbf{CL5}. We further let \textbf{CL5}$^{-}$ denote the system that results from additionally removing the \emph{duplication} rule from CL5.

Let $S$ be one of the cirquent calculus systems \textbf{CCC}, \textbf{CL5} or \textbf{CL5}$^{-}$. A \textbf{proof} of a cirquent $C$ in $S$ is a tree of cirquents whose root is $C$ where each node follows from its children by one of the rules of $S$. When we say a formula $F$ is provable in $S$, we mean the cirquent containing a single oformula $F$ with one ogroup and arc is provable in $S$. In this paper, for simplicity, we are only interested in proving formulas, even though all of our results almost straightforwardly extend from formulas (as special cases of cirquents) to all cirquents.   Correspondingly, we agree that, unless suggested otherwise by the context, ``provability" means ``provability of formulas". It was shown in \cite{ICC} that the provable formulas of \textbf{CCC} coincide exactly with those of classical propositional logic. Interestingly enough, the provable formulas of \textbf{CL5} can also be described in a very natural way. Because we will rely on this result and its associated concepts in several of our later proofs, we shall give the relevant details here.

 A \textbf{substitution} for a formula $C$ is a function $\sigma$ that maps every atom $P$ in $C$ to some formula $\sigma(P)$. If $\sigma(P)$ is an atom for every $P$ in $C$, then we say that $\sigma$ is an \textbf{atomic-level substitution}. This notion can be extended to all formulas by requiring that $\sigma(\neg P) = \neg \sigma(P)$, $\sigma(F \vee G) = \sigma(F) \vee \sigma(G)$ and $\sigma(F \wedge G) = \sigma(F) \wedge \sigma(G)$. Let $A$ and $B$ be formulas. $B$ is said to be an \textbf{instance} of $A$ iff there exists a substitution $\sigma$ such that $\sigma(B) = A$. If $\sigma$ is an atomic level substitution, then $B$ is an \textbf{atomic-level instance} of $A$. A formula is called \textbf{binary} iff no atom has more than two occurrences in it. A binary formula is said to be \textbf{normal} iff, whenever an atom occurs twice in the formula, one occurrence is positive and the other is negative. The following theorem is a combination of Theorem 12 and Lemma 9 of \cite{ICC}. 

\begin{thm}\label{thm:cl5bin} (Japaridze) A formula is provable in CL5 iff it is an instance of a binary tautology iff it is an atomic-level instance of a normal binary tautology.
\end{thm}

At this point, a semantical characterization of the provable formulas of \textbf{CL5}$^-$ has yet to be given. In light of the results of this paper, such a finding would perhaps be very interesting. 

\section{CL5 without duplication is NP-complete}\label{3}

In this section, we show that deciding provability for a formula of the logic \textbf{CL5}$^{-}$ is NP-complete. This result is achieved by a combination of two theorems, the first of which states that every formula $F$ provable in \textbf{CL5}$^{-}$ has a proof whose size is polynomial in the size of $F$. Towards this goal, we begin by presenting a set of technical lemmas.   

\begin{lem}\label{lem:bintautology}
Every formula provable in CL5$^-$ is an atomic-level instance of a normal binary tautology.       
\end{lem}

\begin{proof}
This follows immediately from the ``$\Rightarrow$" direction of Theorem \ref{thm:cl5bin}, given that the theorems of CL5$^-$ form a subset of the theorems of CL5.   
\end{proof}   

\begin{lem}\label{lem:width}
Let $F$ be a formula provable in CL5$^-$ with $n$ positive occurrences of atoms. No CL5$^-$ proof of $F$ can contain a cirquent with more than $n$ ogroups.    
\end{lem}

\begin{proof}
Let $F$ and $n$ be as in the condition of the lemma and let $\Delta$ be a proof tree for $F$. Assume for a contradiction that $\Delta$ has a cirquent $D$ that contains more than $n$ ogroups. Now let $\Delta'$ be the subtree of $\Delta$ rooted at $D$. A careful examination of the rules of CL5$^-$ will show that, no matter which series of rules are applied in $\Delta'$, the total number of ogroups in all the leaves of $\Delta'$ must be greater than or equal to the number of ogroups in $D$. This is because the premise cirquent(s) in the application of a CL5$^-$ rule must, together, contain at least as many ogroups as the conclusion. Every leaf in $\Delta'$ must be derived by either the \emph{identity axiom} or the \emph{empty cirquent axiom} and thus cannot contain more than 1 ogroup. This means $\Delta'$ has at least $n+1$ leaves with a single ogroup, each of which must follow by the \emph{identity axiom}. To be a consequence of the \emph{identity axiom}, a cirquent must contain exactly two oformulas $K$ and $\neg K$. But the total number of positive occurrences of atoms in oformulas in the leaves of $\Delta$ cannot exceed $n$. This is because, by our assumption, the conclusion of $\Delta$ contains a single oformula with $n$ positive literals. Further, the oformulas of the premise cirquent(s) in the application of a CL5$^-$ rule cannot contain more positive occurrences of atoms than the oformulas of the conclusion. This means at most $n$ applications of the \emph{identity axiom} are possible, contradiction.              
\end{proof}  

The \textbf{length of a formula} $F$ is defined as the total number of occurrences of literals and connectives in $F$. The \textbf{size of a group} is the total number of oformulas it contains. The \textbf{size of a cirquent} is then the sum of the lengths of the oformulas in its pool plus the sum of the sizes of each ogroup in its structure. Naturally, the \textbf{size of a proof} is the sum of the sizes of the cirquents it contains.    

\begin{lem}\label{lem:depth}
Let $F$ be a formula provable in CL5$^-$ and let $k$ be the length of $F$. There exists a CL5$^-$ proof of $F$ using $O(k^6)$ applications of rules.  
\end{lem}

\begin{proof}
Let $F$ be a formula and let $\Delta$ be a CL5$^-$ proof tree for $F$. Given $\Delta$, we show that a new proof $\Delta^{\prime}$ can be obtained from $\Delta$ such that $\Delta'$ proves $F$ using $O(k^6)$ applications of rules.

We begin by noting that any proof containing applications of the \emph{empty cirquent axiom} can be transformed into one in which no applications of the rule are made\footnote{The only purpose of this axiom in \cite{ICC} is to ensure the provability of the empty cirquent itself; everything else is provable without using the empty cirquent axiom.}. Indeed, observe that the empty cirquent can only be a (``dummy") premise of \emph{mix}. In such a case, the conclusion of \emph{mix} is simply the same as its other premise. Therefore, the empty cirquent can be deleted and the application of \emph{mix} can be skipped. This can be done for all empty cirquents contained in the proof until none remain. Let $\Delta_1$ be the proof that results from removing all applications of the \emph{empty cirquent axiom} in the preceding manner. 

Observe now that $\Delta_1$ must contain exactly one application of \emph{conjunction introduction} or \emph{disjunction introduction} for each occurrence of $\wedge$ or $\vee$ in $F$, respectively. This is because every application of either rule introduces, in its conclusion, a single $\wedge$ or $\vee$ connective that cannot be later removed by the application of any CL5$^-$ rule. The number of positive occurrences of atoms in $F$ also bounds the number of applications of the \emph{identity axiom}. This is because, in a bottom up view of a proof, no CL5$^-$ rule allows formulas to be removed or merged. With the number of applications of \emph{conjunction introduction}, \emph{disjunction introduction} and the \emph{identity axiom} all bounded by $O(k)$ in $\Delta_1$, we can also bound the number of applications of the \emph{mix} rule. To see this, first note that because $\Delta_1$ contains no applications of the \emph{empty cirquent axiom}, every leaf node must be derived by the \emph{identity axiom}. That is, there are $O(k)$ paths from root to leaf in $\Delta_1$. Viewing the proof tree in a top down fashion, \emph{mix} is the only CL5$^-$ rule which allows the proof tree to ``branch", given it is the only rule with more than one premise.  Thus, each application of \emph{mix}, which must take two premises, increases the number of paths from root to leaf by 1. Because the number of such paths is bounded by $O(k)$, so too must be the number of applications of the \emph{mix} rule.  

Each application of \emph{pool weakening} introduces, in its conclusion, an oformula that must be present (perhaps only as a proper subformula of some oformula) in the conclusion of $\Delta_1$. This is because no CL5$^-$ rule can remove, in its conclusion, an oformula contained in its premise cirquent. Thus we have an immediate bound of $O(k)$ for the number of applications of \emph{pool weakening} in $\Delta_1$. Observe that this also implies a bound of $O(k)$ on the number of oformulas contained in any cirquent. To bound the number of possible applications of \emph{ogroup weakening} in $\Delta_1$, notice that no cirquent in $\Delta_1$ can contain more than $k^2$ arcs. This is because, by Lemma \ref{lem:width}, $k$ bounds the maximum number of ogroups in any cirquent and, as noted previously, $O(k)$ bounds the number of oformulas in any cirquent. Further, \emph{conjunction introduction} and \emph{disjunction introduction} are the only CL5$^-$ rules whose conclusion can contain fewer arcs than its premise. The number of applications of these rules in $\Delta_1$ does not exceed $O(k)$ and each application can remove no more than the maximum $k^2$ arcs contained in any premise cirquent. Thus, no more than $O(k^3)$ total arcs can be removed from premise to conclusion for all applications of \emph{conjunction} and \emph{disjunction introduction} in $\Delta_1$. Each application of \emph{ogroup weakening} creates a single arc in its conclusion that was not present in its premise. As the conclusion of $\Delta_1$ contains a single arc, no more than $O(k^3)$ arcs can be created from premise to conclusion in $\Delta_1$ by applications of \emph{ogroup weakening}.

Every leaf of $\Delta_1$ must be derived by the \emph{identity axiom}, which is applied at most $O(k)$ times in $\Delta_1$. That is, there are at most $O(k)$ paths from root to leaf in $\Delta_1$. Each such path $\Delta_1^i$ for $1 \leq i \leq O(k)$ contains at most $O(k^3)$ applications of CL5$^-$ rules other than \emph{oformula} or \emph{ogroup exchange}. Thus, any $\Delta_1^i$ can contain at most $O(k^3)$ sequences $\Phi_w$ for $0 \leq w \leq O(k^3)$ where each $\Phi_w := C_1^w, C_2^w, ... , C_n^w$ (for $n \in \mathbb{N}$) is a sequence of cirquents such that every cirquent $C_i^w$ follows from $C_{i+1}^w$ by \emph{exchange} (of either sort) for all $0 \leq i < n$. In $\Delta_1$, each $\Phi_w$ can be of arbitrary length. We show, however, that a new proof $\Delta'$ can be obtained from $\Delta_1$ such that each $\Phi_w$ uses no more than $O(k^2)$ applications of \emph{oformula} and \emph{ogroup exchange}. This is because, the last cirquent $C_n^w$ in any $\Phi_w$ can always be obtained from $C_1^w$ using no more than $O(k^2)$ applications of \emph{oformula} and \emph{ogroup exchange}. By lemma \ref{lem:width}, the first cirquent $C_1^w$ can contain at most $k$ ogroups. Moving all of these ogroups to an arbitrary position in $C_n^w$ can be done with $O(k^2)$ applications of \emph{ogroup exchange}. Additionally, moving each of $O(k)$ possible oformulas in $C_1^w$ to an arbitrary position in $C_n^w$ also requires no more than $O(k^2)$ applications of \emph{oformula exchange}. This means that every sequence $\Phi_w$ in $\Delta_1$ can be replaced with a new sequence $\Phi'_w$ such that $\Phi'_w$ uses no more than $O(k^2)$ applications of \emph{oformula} or \emph{ogroup exchange}. Let $\Delta'$ be the CL5$^-$ proof tree that results from replacing each $\Phi_w$ in $\Delta_1$ by $\Phi'_w$. $\Delta'$ must again contain $O(k)$ paths from root to leaf. Each such path now has a bound of $O(k^3) \times O(k^2)$ applications of \emph{oformula} or \emph{ogroup exchange}, resulting in a total bound of $O(k^6)$ applications of \emph{exchange} in all branches of $\Delta'$.  
 
It should be clear that if $\Delta$ proves $F$, then $\Delta'$ will also prove $F$. The number of applications of non-\emph{exchange} rules does not change from $\Delta_1$ to $\Delta'$ and remains $O(k^3)$. Further, $\Delta'$ uses no more than $O(k^6)$ applications of \emph{exchange}. Totaling the number of applications of every rule in $\Delta'$, we obtain a bound of $O(k^3) + O(k^6)$.  
\end{proof}  

\begin{thm}\label{NP}
Let $F$ be a formula provable in CL5$^-$ and let $k$ be the size of $F$. There exists a CL5$^-$ proof of $F$ whose size is polynomial in $k$. 
\end{thm}

\begin{proof} 
By lemma \ref{lem:depth}, there exists a CL5$^-$ proof $\Delta$ of $F$ that uses a polynomial number of rule applications. Further, by lemma \ref{lem:width}, the size of  any cirquent in $\Delta$ must be polynomial in $k$. Thus, the maximum size of each cirquent multiplied by the maximum number of cirquents in $\Delta$ yields a polynomial bound on the size of $\Delta$.    
\end{proof}

The result of Theorem \ref{NP} effectively places the provability problem for CL5$^-$ formulas in NP. We further this result by additionally showing that the problem is NP-complete. Before giving the proof, we solidify some standard concepts that will be used within it. A \textbf{graph} $G = (V, E)$ is an ordered pair made up of a set of vertices ($V$) and edges ($E$), with each edge being an unordered pair of vertices. The \textbf{degree} of a vertex $v$ in $G$, denoted deg(v), is defined as the number of edges incident to $v$. Given a graph $G$, a \textbf{vertex cover} is $U \subseteq V$ such that every edge of $E$ is incident to at least one vertex in $V$. In complexity theory, the \textbf{vertex cover problem} can be stated as a decision problems as follows. Given a graph $G$ and number $k$, does $G$ have a $k$ vertex cover, i.e. a vertex cover using at most $k$ vertices? It is well known that this problem is NP-complete.

\begin{thm}
Deciding provability for the logic CL5$^-$ is NP-complete.
\end{thm}   

\begin{proof}
It follows from Theorem \ref{NP} that CL5$^-$ provability is in NP. To see it is NP-hard, we give a polynomial time reduction $f$ to it from the vertex cover problem. Fix some arbitrary graph $G = (V, E)$ and some $k \in \mathbb{N}$. The reduction, borrowed from \cite{DPLL}, follows.  
$$f(V, E, k) := (\Psi(k)) \vee (\Theta(V,E)) \vee (\Omega(E))$$

$$\Psi(k) := \underbrace{q \vee q \vee ... \vee q}_
{\mbox{total of $k$ literals}}$$

$$\Theta(V,E) := \underbrace{(\neg q \wedge \underbrace{(\neg v_1 \vee \neg v_1 \vee ... \vee \neg v_1) }_
{\mbox{deg($v_1$)}} ) \vee ... \vee(\neg q \wedge \underbrace{(\neg v_n \vee \neg v_n \vee ... \vee \neg v_n) }_
{\mbox{deg($v_n$)}})}_
{\mbox{for each vetex  $v_1,v_2,...,v_n \in$ V}}$$

$$\Omega(E) := \underbrace{(e_1^1 \vee e_1^2) \wedge (e_2^1 \vee e_2^2) \wedge ... \wedge (e_m^1 \vee e_m^2)}_
{\mbox{for each edge $e_1,e_2,...,e_m \in$ E, }} $$     

Above, $q$ is a new atom that differs from all other atoms in the formula. In $\Omega(E)$, $e_i^1$ and $e_i^2$ are the vertices on the endpoints of edge $e_i$ for $1 \leq i \leq m$. It is also important to note that each $v_j$ for $1 \leq j \leq n$ represents both an atom in $f(V, E, k)$ as well as the label of a vertex in $G$.

Obviously the mapping $f$ is computable in polynomial time. It now remains to show that a graph $(V,E)$ has a $k$ vertex cover if and only if the formula $f(V,E,k)$ is provable in CL5$^-$.

``$\Rightarrow$" Our reduction is identical to that of \cite{DPLL} (Section 5.2) for multiplicative affine (direct) logic. As a direct consequence of Theorem 3 from \cite{ICC}, CL5$^-$ proves every formula provable in multiplicative affine logic\footnote{From Theorem 3 of \cite{ICC}, we have affine logic = CL5$^*$, where  CL5$^*$ is the system CL5 with the limitation that cirquents contained in proofs do not have groups that share oformulas. Because \emph{upward} (resp. \emph{downward}) \emph{duplication} is only applicable when the premise (resp. conclusion) of the rule is a cirquent in which groups share oformulas, all cirquents provable in CL5$^*$ are also provable in  CL5$^-$.  }, so the result follows immediately. 

``$\Leftarrow$" Assume $f(V,E,k)$ is provable in CL5$^-$. We need to show that there exists a $k$ vertex cover of $G = (V,E)$. By Theorem \ref{thm:cl5bin}, there exists a normal binary tautology $\beta$ and an atomic-level substitution $\sigma$ such that $\sigma(\beta) = f(V,E,k)$. Let $\Psi'$ be the subformula of $\beta$ such that $\sigma(\Psi') = \Psi(k)$. Let $\Pi$ be the set of the $k$ atoms of $\Psi'$, and let $\Sigma \subseteq \Pi$ be the set of those members of $\Pi$ that have (not only positive but also) negative occurrences of $\beta$. For each $l \in \Sigma$, let $\Theta_l$ be the conjunct in the $\Theta(V,E)$ component of $f(V,E,k)$ that contains $\neg \sigma(l)$. By our construction, each $\Theta_l$ takes the form $$(\neg q \wedge (\neg v_j \vee \neg v_j \vee ... \vee \neg v_j))$$ for $1 \leq j \leq n$. Define a vertex cover $V' \subset V$ for $G$ as the set of all vertices in $G$ labeled $v_j$ where $v_j$ is contained in $\Theta_l$ for some $l \in \Sigma$. 

To see that $V'$ is indeed a $k$ vertex cover for $G$, note first that $\sigma$ is an atomic-level substitution, and hence each positive literal $l \in \Sigma$ must be mapped to a unique positive occurrence of $q$ in $f(V,E,k)$. By the definition of $\Psi(k)$, however, there are exactly $k$ positive occurrences of $q$ in $f(V,E,k)$. This means $\Sigma$ contains no more than $k$ atoms and subsequently the vertex cover derived from $\Sigma$ contains no more than $k$ vertices.  
  
 Because $f(V,E,k)$ is an atomic-level instance of $\beta$, both formulas take exactly the same form, the only difference being in the atoms they are built from. Let us now define a model $^*$ for $\beta$. For each $l \in \Pi$ let $l^* =$ $\perp$. For any atom $l \not\in \Pi$ with the property that $\sigma(l) = q$, let $l^* =$ $\top$. Further, for each atom $l$ of $\beta$ where $\sigma(l) = a$ for some $a \in V'$, let $l^* =$ $\top$. All remaining atoms $l$ of $\beta$ should be interpreted as $l^* =$ $\perp$.      
 
 Note that every atom $l$ of the earlier defined $\Phi'$ is interpreted as $l^* = \perp$, meaning that $(\Phi')^* = \perp$. Next let $\Theta'$ be the subformula of $\beta$ such that $\sigma(\Theta') = \Theta(V,E)$. Each disjunct of $\Theta'$ takes the form $\neg l_1 \wedge (\neg l_2 \vee \neg l_3 \vee ... \vee \neg l_n)$ for some positive integer $n$. If $l_1 \not\in \Sigma$ then $l_1^* =$ $\top$ and $\neg (l_1^*) =$ $\perp$, making the entire subformula evaluate to $\perp$ under $^*$. If $l_1 \in \Sigma$ then $l_1^* =$ $\perp$ and $\neg(l_1^*) =$ $\top$. Note, however, that when $l_1 \in \Sigma$ then $\sigma(l_m) = a$ where $a \in V'$ and $m \in \{2, 3, \ldots \}$. This means $l_m^* =$ $\top$ and $\neg (l_m^*) =$ $\perp$, again making the subformula evaluate to $\perp$ under $^*$. Because each disjunct of $\Theta'$ evaluates to $\perp$ under $^*$, we have $(\Theta')^* =$ $\perp$. 

Finally, let $\Omega'$ be the subformula of $\beta$ such that $\sigma(\Omega') = \Omega(E)$. Because $\beta = (\Psi') \vee (\Theta') \vee (\Omega')$ is a tautology where $(\Psi')^* =$ $\perp$ and $(\Theta')^* =$ $\perp$, we have $(\Omega')^* =$ $\top$. Each conjunct of $\Omega'$ takes the form $l_1 \vee l_2$ such that $\sigma(l_1)$ and $\sigma(l_2)$ are the endpoints of an edge in $E$. Further, by our construction, every edge in $E$ is represented by such a conjunct in $\Omega'$. Observer that, by our definition of $^*$, we have $l^* =$ $\top$ for an atom $l$ of $\Omega'$ only when $\sigma(l) \in V'$. Thus, $\Omega'$ is true only when every edge in $E$ has an endpoint in $V'$. That is, $V'$ is a vertex cover of $G$.    
\end{proof}

\section{CL5 is $\Sigma_2^p$-complete}

This section contains our main result, namely, the $\Sigma_2^p$-completeness of the decision problem for provability of formulas in \textbf{CL5}. 

\begin{lem}\label{lem:ptime_algorithm}
Deciding provability for CL5 is in $\Sigma_2^p$. 
\end{lem}  

\begin{proof}
The following is a $\Sigma_2^p$ algorithm that, if view of Theorem \ref{thm:cl5bin}, decides provability of a formula $G$ in CL5. On input $G$, existentially guess a binary formula $F$ such that $G$ is an instance of $F$. Then, universally guess a truth assignment $^*$ for $F$. If $F$ is true under $^*$, accept. Otherwise, reject. 
\end{proof}

Let TQBF-$\Sigma_2$ be the problem of deciding truth for a quantified Boolean formula of the form $\exists X \forall Y \Theta$, where $X$ and $Y$ are sequences of variables and $\Theta$ is a quantifier-free Boolean formula all of whose variables are among $X$ or $Y$. As shown in Theorem 4.1 of \cite{PTIME}, this problem is $\Sigma_2^p$-complete. We say an atom $z$ is \textbf{isolated} in a formula $\phi$ if there is only a single (positive or negative) occurrence of $z$ in $\phi$. Otherwise it is \textbf{non-isolated}.  

\begin{thm}
Deciding provability for the logic CL5 is $\Sigma_2^p$-complete. 
\end{thm}  

\begin{proof}
By lemma \ref{lem:ptime_algorithm}, the problem of deciding provability for CL5 is in $\Sigma_2^p$. To show it is $\Sigma_2^p$-hard, we construct a polynomial time mapping reduction $f$ from TQBF-$\Sigma_2$ to CL5-provability. We can safely restrict our attention to instances of TQBF-$\Sigma_2$ where the Boolean portion of the formula is in disjunctive normal form, as the complexity of the problem is not reduced with these limitations. Our reduction follows.

Let $\phi$ be a formula that takes the form $\exists X \forall Y \Theta$ with $X$ and $Y$ being sets of variables and $\Theta$ being a boolean formula in disjunctive normal form all of whose variables are among $X \cup Y$. The following steps are used to construct the corresponding CL5 formula $f(\phi)$.

\begin{enumerate}

	\item For each $z \in X$, let $k^z$ and $t^z$ be the number of positive occurrences of the literals $z$ and $\neg z$ in $\Theta$, respectively. Define $$g(z) = Z^z \wedge (Z^z \rightarrow \underbrace{u_1^z \wedge u_2^z \wedge ... \wedge u_{k^z}^z }_{\mbox{$k^z$ literals}}) \wedge (Z^z \rightarrow \underbrace{\neg v_1^z \wedge \neg v_2^z \wedge ... \wedge \neg v_{t^z}^z}_{\mbox{$t^z$ literals}})$$ where $Z^z$, $u_1^z, u_2^z, ... , u_{k^z}^z$ and $v_1^z, v_2^z, ..., v_{t^z}^z$ are all fresh\footnote{Here, fresh variables are those not occurring elsewhere in $\phi$.} pairwise distinct variables, unique\-ly chosen for $z$. Here, if $k^z$ (resp. $t^z$) is 0, the second (resp. third) conjunct should be omitted. Now let $z_1, z_2, ... , z_l = X$ and let $\phi_0$ be the formula $(g(z_1) \wedge g(z_2) \wedge ... \wedge g(z_l)) \rightarrow \Theta_1$. Here $\Theta_1$ is a formula derived from $\Theta$ such that, for each $z \in X$, every positive occurrence of $z$ in $\Theta$ is replaced by a unique literal from $u_1^z, u_2^z, ... , u_k^z$ and every occurrence of $\neg z$ in $\Theta$ is replaced by a unique literal from $\neg v_1^z, \neg v_2^z, ..., \neg v_t^z$. 

\item Consider any $y \in Y$. Let $r_y$ be the number of positive occurrences of $y$ in $\phi_0$, and $s_y$ be the number of negative occurrences. For each pair $i,j$ with $1\leq i\leq r_y$ and $1\leq j\leq s_y$, we choose a fresh and unique variable $P^{y}_{i,j}$. Now, define $f(\phi)$ to be the result of replacing in $\phi_0$, for each $y \in Y$, every positive occurrence of the literal $y$ by $(P_{i,1}^y \vee P_{i,2}^y \vee ... \vee P_{i,s_y}^y)$ and every (positive) occurrence of the literal $\neg y$ by $(\neg P_{1,j}^y \vee \neg P_{2,j}^y \vee ... \vee \neg P_{r_y,j}^y)$ where $1 \leq i \leq r_y$ (resp. $1 \leq j \leq s_y$) is unique for each replacement of an occurrence of $y$ (resp. $\neg y$).\medskip
\end{enumerate}

\noindent``$\Rightarrow$" Assume $\phi$ is true. Then there exists some truth assignment $^{\circ}: X \rightarrow \{\top, \perp\}$ such that $\Theta^*$ is true under any truth assignment $^*$ that extends $^{\circ}$ to $X \cup Y$. By our assumption, $\Theta$ is in disjunctive normal form and hence takes the form $$(\psi_1 \wedge \psi_2 \wedge ... \wedge \psi_n) \vee ... \vee (\psi_{m+1} \wedge \psi_{m+2} \wedge ... \wedge \psi_{m+l})$$ where each $\psi_w$ is either a positive or negative literal. Now let $\Omega$ and $\Sigma$ represent the antecedent and consequent of the outermost implication in $f(\phi)$ such that $f(\phi) = \Omega \rightarrow \Sigma$. Our construction guarantees that $\Sigma$ takes the form $$(\Psi_1 \wedge \Psi_2 \wedge ... \wedge \Psi_n) \vee ... \vee (\Psi_{m+1} \wedge \Psi_{m+2} \wedge ... \wedge \Psi_{m+l})$$ where each  $\Psi_w$ is a disjunction of literals (such a ``disjunction" may have only a single ``disjunct"). That is, $\Sigma$ is obtained from $\Theta$ by replacing each oliteral $\psi_w$ by $\Psi_w$, where $\Psi_w$ is a disjunction of literals. We will henceforth use $\psi_w$ to represent a unique position in $\Theta$ and $\Psi_w$ to represent the corresponding position in $\Sigma$.

We want to show that $f(\phi)$ is provable in CL5. By Theorem \ref{thm:cl5bin}, it suffices to show that $f(\phi)$ is an instance of a binary tautology. We construct a binary tautology $\Phi$ of which $f(\phi)$ is an instance. Namely, we let $\Phi$ be the formula obtained from $f(\phi)$ as follows. For each $z \in X$, if $z^{\circ} = \top$ (resp. $\perp$) replace the third (resp. second) occurrence of $Z^z$ by an atom $Q^z$ such that $Q^z$ does not occur elsewhere in $\Phi$ and is unique for each $z \in X$. It should be easy to see that $\Phi$ is a quantifier free binary formula and $f(\phi)$ is an instance of $\Phi$. We need only show that $\Phi$ is a tautology. Again notice that $\Phi$ takes form $\Phi = \Pi \rightarrow \Sigma$ where $\Sigma$ is the same as in $f(\phi)$. 

Given $X = z_1, z_2, ... , z_l$, we have $\Pi = g(z_1)' \wedge g(z_2)' \wedge ... \wedge g(z_l)'$ where each $g(z)'$	matches one of the following forms\footnote{It is possible that the third (resp. second) conjunct in formula 1 (resp. 2) is absent.}.

\begin{enumerate}
\item $Z^z \wedge (Z^z \rightarrow u_1^z \wedge u_2^z \wedge ... \wedge u_{k^z}^z ) \wedge (Q^z \rightarrow \neg v_1^z \wedge \neg v_2^z \wedge ... \wedge \neg v_{t^z}^z)$

\item $Z^z \wedge (Q^z \rightarrow u_1^z \wedge u_2^z \wedge ... \wedge u_{k^z}^z ) \wedge (Z^z \rightarrow \neg v_1^z \wedge \neg v_2^z \wedge ... \wedge \neg v_{t^z}^z)$

\end{enumerate}

For any $g(z)'$ in $\Pi$ where $z \in X$, if $(Z^z)^{\star} =$ $\perp$ then $\Pi^{\star} =$ $\perp$, and hence $\Phi^{\star} = \top$. If $(Z^z)^{\star} =$ $\top$, then all of the literals in the consequent of the implication in $g(z)'$ with antecedent literal $Z^z$ must be true under $^{\star}$, otherwise we will have $\Pi^{\star} =$ $\perp$ and again $\Phi^{\star} =$ $\top$. Thus, we need only guarantee $\Sigma^{\star} = \top$ under truth assignments such that, for each conjunct $g(z)'$ in $\Pi$, $Z^z$ and every literal in the consequent of the implication containing antecedent literal $Z^z$ are true under $^{\star}$. 

Let $\Theta'$ be the formula that results from replacing in $\Theta$ every positive occurrence of $z$ or $\neg z$, where $z \in X$, by its truth value under $^{\circ}$. For example, if $z^{\circ} =$ $\perp $, replace $\neg z$ by $\top$ and $z$ by $\perp$. It should be easy to see that $\Theta'$ is a tautology. For any position $\psi_w$ that contains $z$ (resp. $\neg z$) in $\Theta$ and $\top$ in $\Theta'$, the position $\Psi_w$ in $\Sigma$ must contain a single literal $u_i^z$ (resp. $\neg v_i^z$) such that $(u_i^z)^{\star} = \top$ (resp. $(\neg v_i^z)^{\star} = \top$) if $\Pi^{\star} = \top$. This is because, by our construction, if $z^{\circ} =$ $\top$ (resp. $(\neg z)^{\circ} =$ $\top$), every $u_i^z$ for $1 \leq i \leq k^z$ (resp. $\neg v_i^z$ for $1 \leq i \leq t^z$) must occur in the consequent of the implication in $\Pi$ with non-isolated antecedent $Z^z$. For reasons already discussed, such a literal must be true under any truth assignment $^{\star}$ where $\Pi^{\star} = \top$. Define $\Sigma^{\prime}$ as the formula such that, for each occurrence of $\top$ or $\perp$ in position $\psi_w$ of $\Theta'$, the literal in position $\Psi_w$ of $\Sigma$ is replaced by the same value ($\top$ or $\perp$) as $\psi_w$. Note that $\Sigma^{\prime}$ only substitutes the logical atom in position $\Psi_w$ of $\Sigma$ when $\Psi_w$ contains a positive or negative occurrence of an atom in $X$. Since our goal is to show that $\Sigma$ is true under truth assignments that make $\Pi$ true, we need only show that $\Sigma^{\prime}$ is a tautology.

For a contradiction, assume $\Sigma^{\prime}$ is not a tautology. Then there exists some truth assignment $^{\dagger}$ defined on the variables of $\Sigma^{\prime}$ such that $(\Sigma^{\prime})^{\dagger} =$ $\perp$. We define a truth assignment $^\ddagger$ for $\Theta'$ as follows. If a variable $y \in Y$ is such that, for some $i$ with $1\leq i\leq r_y$, $(P^{y}_{i,1}\vee\ldots\vee P^{y}_{i,s_y})^\dagger =\bot$, we let $y^\ddagger=\bot$; otherwise we let $y^\ddagger=\top$. It is not hard to see that, if a subformula of $\Sigma'$ in a position $\Psi_w$ is false under $^\dagger$, then the subformula of $\Theta'$ in the corresponding position $\psi_w$ is false under $^\ddagger$ (but not necessarily vice versa).  This, in view of the monotonicity of $\vee$ and $\wedge$,  obviously implies that $(\Theta')^\ddagger =\bot$, because $\Sigma'$ and $\Theta'$ have the same $(\vee,\wedge)$-structures. Now we are dealing with a contradiction, because the tautological $\Theta'$ cannot be false.\medskip 

\noindent ``$\Leftarrow$" Assume $f(\phi)$ is provable in CL5. By Theorem \ref{thm:cl5bin}, for some normal binary tautology $\Phi$ there exists an atomic level substitution $\sigma$ such that $\sigma(\Phi)$ $=$ $f(\phi)$. Let $\Phi = \Pi \rightarrow \Sigma$, where $\Pi$ and $\Sigma$ represent the antecedent and consequent of the outermost implication in $\Phi$. We want to define a truth assignment $^{\circ}: X \rightarrow \{\top, \perp\}$ such that $\Theta^*$ is true for any truth assignment $^*$ that extends $^{\circ}$ to $X \cup Y$. In what follows, we define such a partial truth assignment $^{\circ}$ for $\Theta$ while concurrently defining a partial truth assignment $^{\bullet}$ for $\Phi$. 

\emph{Procedure 1} - Consider a $z \in X$ and let $A$ be the atom of $\Phi$ such that the first occurrence of $Z^z$ in $f(\phi)$ {\bf originates} from A (i.e., A gets replaced by $\sigma(A)=Z^z$) when transitioning from $\Phi$ to $\sigma(\Phi)=f(\phi)$. 

\emph{Case 1}: The second and third occurrences of $Z^z$ in $f(\phi)$ originate from $A$ and $B$, respectively (for some $B \neq A$ in $\Phi$). Define $z^{\circ} = \top$, $A^{\bullet} = \top$ and $B^{\bullet} = \perp$. The consequent\footnote{Here and later in similar contexts, the ``consequent of W" should be understood as the consequent of the implication whose antecedent is W. Similarly for ``antecedent of W".} of the second occurrence of $A$ in $\Phi$ should take the form $a_1 \wedge a_2 \wedge ... \wedge a_{k^z}$ for some positive literals $a_1, a_2, ..., a_{k^z}$ in $\Phi$ and the consequent of $B$ should take the form $\neg b_1 \wedge \neg b_2 \wedge ... \wedge \neg b_{t^z}$ for some negative literals $\neg b_1, \neg b_2, ..., \neg b_{t^z}$. Let all of $a_1, a_2, ..., a_{k^z}$ and $b_1, b_2, ..., b_{t^z}$ be true under $^{\bullet}$. 

\emph{Case 2}: The second and third occurrences of $Z^z$ in $f(\phi)$ originate from $B$ and $A$, respectively (for some $B \neq A$ in $\Phi$). Define $z^{\circ} = \perp$, $A^{\bullet} = \top$ and $B^{\bullet} = \perp$. The consequent of the second occurrence of $A$ in $\Phi$ should take the form $\neg b_1 \wedge \neg b_2 \wedge ... \wedge \neg b_{t^z}$ for some negative literals $\neg b_1, \neg b_2, ..., \neg b_{t^z}$ and the consequent of $B$ should take the form $a_1 \wedge  a_2 \wedge ... \wedge  a_{k^z}$ for some positive literals $ a_1,  a_2, ..., a_{k^z}$. Let all of $a_1, a_2, ..., a_{k^z}$ and $b_1, b_2, ..., b_{t^z}$ be false under $^{\bullet}$.  

\emph{Case 3}: There are only two occurrences of $Z^z$ in $f(\phi)$, both of which originate from $A$. If the consequent of the second occurrence of $A$ in $\Phi$ takes the form $a_1 \wedge a_2 \wedge ... \wedge a_{k^z}$ for some positive literals $a_1, a_2, ..., a_{k^z}$ then define $z^{\circ} = \top$ and $A^{\bullet} = \top$. Further, let all of $a_1, a_2, ..., a_{k^z}$ be true under $^{\bullet}$. If the consequent of the second occurrence of $A$ in $\Phi$ takes the form $\neg b_1 \wedge \neg b_2 \wedge ... \wedge \neg b_{t^z}$ for some negative literals $\neg b_1, \neg b_2, ..., \neg b_{t^z}$ then define $z^{\circ} = \perp$ and $A^{\bullet} = \top$. Further, let all of $b_1, b_2, ..., b_{t^z}$ be false under $^{\bullet}$.

\emph{Case 4}:  If none of the cases $1-3$ are satisfied, then the second and third occurrences of $Z^z$ in $f(\phi)$ must originate from $B$ and $C$, respectively (where $A \neq B$ and $A \neq C$). Define  $z^{\circ} =$ $\perp$, $A^{\bullet} = \top$ and $B^{\bullet} = C^{\bullet} =$ $\perp$. Further, let every atom in the consequent of $B$ evaluate to $\perp$ under $^{\bullet}$ and every atom in the consequent of $C$ evaluate to $\top$ under $^{\bullet}$.

As we remember, $\Theta$ in disjunctive normal form, taking the form $$(\psi_1 \wedge \psi_2 \wedge ... \wedge \psi_n) \vee ... \vee (\psi_{m+1} \wedge \psi_{m+2} \wedge ... \wedge \psi_{m+l})$$ where each $\psi_w$ is a literal. It is also the case that $\Sigma$ takes the form $$(\Psi_1 \wedge \Psi_2 \wedge ... \wedge \Psi_n) \vee ... \vee (\Psi_{m+1} \wedge \Psi_{m+2} \wedge ... \wedge \Psi_{m+l})$$ where each  $\Psi_w$ is a disjunction of literals (possibly with just a single ``disjunct"). Thus, we will use our previous convention in which $\psi_w$ represents a unique position in $\Theta$ and $\Psi_w$ represents the corresponding position in $\Sigma$.

Our construction guarantees that for each position $\psi_w$ in $\Theta$ containing some positive (resp. negative) literal $z$ (resp. $\neg z$) such that $z \in X$, $\Psi_w$ contains a positive (resp. negative) literal $a$ (resp. $\neg a$). By \emph{procedure 1} we have $z^{\circ} = a^{\bullet}$. As in the previous direction, let $\Theta'$ be the formula that results from replacing, in $\Theta$, for every $z \in X$, all positive occurrences of the literals $z$ and $\neg z$ by their truth values under $^{\circ}$. Let $\Sigma'$ be the formula such that, for each position $\psi_w$ in $\Theta'$ containing a logical atom $d \in\{ \top,  \bot\}$, the literal in the corresponding position $\Psi_w$ of $\Sigma$ is replaced by $d$. The interpretation $^{\bullet}$ defined as part of \emph{procedure 1} is such that $\Pi^{\bullet} = \top$. Because $\Phi$ is a tautology, any extension of $^{\bullet}$ defined on all atoms of $\Sigma$ must make $\Sigma$ true. We also know that for each position $\Psi_w$ in $\Sigma'$ that contains a logical atom $\top$ (resp. $\perp$), the literal in position $\Psi_w$ of $\Sigma$ evaluates to $\top$ (resp. $\perp$) under $^{\bullet}$. This means that $\Sigma'$ is true regardless of how its non-logical atoms are interpreted and is a tautology. Our goal is to show that $\Theta'$ is a tautology as well. 

Pick some arbitrary truth assignment $^*$ for $\Theta'$. We show $(\Theta')^{*} = \top$. For each position $\psi_w$ in $\Theta'$ containing a positive literal $y$, the corresponding position $\Psi_w$ in $\Sigma'$ contains a disjunction of positive literals ($q_{i,1}^y \vee q_{i,2}^y \vee ... \vee q_{i,s_y}^y$) such that $\sigma(q_{i,1}^y) = P_{i,1}^y$, $\sigma(q_{i,2}^y) = P_{i,2}^y$ $...$, $\sigma(q_{i,s_y}^y) = P_{i,{s_y}}^y$ for some $1 \leq i \leq r_y$. Similarly, for each position $\psi_w$ in $\Theta'$ containing a negative literal $\neg y$, the corresponding position $\Psi_w$ in $\Sigma'$ contains a disjunction of negative literals ($\neg t_{1,j}^y \vee \neg t_{2,j}^y \vee ... \vee \neg t_{r_y,j}^y$) such that $\sigma(t_{1,j}^y) = P_{1,j}^y$, $\sigma(t_{2,j}^y) = P_{2,j}^y$, $...$, $\sigma(t_{r_y}^y) =  P_{{r_y,j}}^y$ for some $1 \leq j \leq s_y$. We now define a truth assignment $^{\dagger}$ for $\Sigma'$. If $y^* = \top$, let $^{\dagger}$ be such that all of the corresponding literals $\neg t_{1,j}^y, \neg t_{2,j}^y, ..., \neg t_{r_y,j}^y$, for every $1 \leq j \leq s_y$ evaluate to false under $^{\dagger}$. If $y^* = \perp$, let $^{\dagger}$ be such that all corresponding literals $q_{i,1}^y, q_{i,2}^y, ..., q_{i,s_y}^y$, for every $1 \leq i \leq r_y$ evaluate to false under $^{\dagger}$. With some thought, one can see that whenever $\Theta'$ has a false (under $^*$) literal in a position $\psi_w$, $\Sigma'$ has a false (under $^\dagger$) disjunction of literals in the corresponding position $\Psi_w$. So, if $(\Theta')^*=\bot$, then $(\Sigma')^\dagger=\bot$. That is, $(\Sigma')^\dagger=\top$ (which is the case due to the tautologicity of $\Sigma'$) implies $(\Theta^*)=\top$.
\end{proof}

\end{document}